\documentclass[
  aps,
  prx,
  reprint,
  superscriptaddress,
  longbibliography,
  floatfix,nofootinbib
]{revtex4-2}

\usepackage{amsmath,amssymb,amsthm,mathtools}
\usepackage{bm}
\usepackage{booktabs}
\usepackage{enumitem}
\usepackage{float}
\usepackage{microtype}
\usepackage{xcolor}
\usepackage{tikz}
\usepackage[normalem]{ulem}
\usetikzlibrary{arrows.meta,positioning,fit,shapes.misc}
\usepackage{xurl}
\usepackage[
  breaklinks=true,
  colorlinks=true,
  linkcolor=blue,
  citecolor=blue,
  urlcolor=magenta
]{hyperref}

\newtheorem{theorem}{Theorem}[section]
\newtheorem{lemma}[theorem]{Lemma}

\newtheorem{corollary}[theorem]{Corollary}
\theoremstyle{definition}
\newtheorem{definition}[theorem]{Definition}
\newtheorem{problem}{Problem}
\theoremstyle{remark}
\newtheorem{remark}[theorem]{Remark}

\newcommand{\C}{\mathbb{C}}
\newcommand{\U}{\mathrm{U}}
\newcommand{\ii}{\mathrm{i}}
\newcommand{\eps}{\epsilon}
\newcommand{\ket}[1]{\lvert #1\rangle}
\newcommand{\bra}[1]{\langle #1\rvert}
\newcommand{\proj}[1]{\lvert #1\rangle\!\langle #1\rvert}
\newcommand{\norm}[1]{\left\lVert #1\right\rVert}
\newcommand{\SELECT}{\operatorname{SELECT}}

\DeclareMathOperator{\diag}{diag}

\newcommand{\B}{\mathsf{B}}
\newcommand{\X}{\mathsf{X}}
\newcommand{\Y}{\mathsf{Y}}
\newcommand{\f}{\mathsf{f}}
\newcommand{\A}{\mathsf{A}}

\begin{document}

\title{Quantum Circuits for General Unitaries: Improved \(T\)-Count via Block Flattening and Dilation}

\author{Pei Yuan}
\email{peiyuan0104@gmail.com}

\author{Shengyu Zhang}
\email{shengyuzhang@gmail.com}

\author{Wei Zi}
\email[Quantum Science Center of Guangdong-Hong Kong-Macao Greater Bay Area, Shenzhen, China. Corresponding author: ]{ziwei.quantum@outlook.com}

\begin{abstract}
Synthesizing arbitrary $n$-qubit unitaries using as few non-Clifford gates as possible is a central problem in fault-tolerant quantum compilation. We present a Clifford+\(T\) circuit construction that approximates any effectively specified unitary to error at most \(\eps\) and whose worst-case \(T\)-count has leading exponential dependence \(2^{5n/4}\) whenever \(\log(1/\eps)=\operatorname{poly}(n)\).  This improves upon the best previous $2^{4n/3}$ scaling.  The key innovation lies in treating the target unitary as a single block-encoded object rather than a long product of simpler operations.
A technique of \textit{block flattening} controls the normalization while preserving an efficient implementation of the block encoding; subsequently, quantum singular value transformation maps its common singular value to one, thereby
recovering the target unitary.
\end{abstract}

\maketitle

\section{Introduction}
\label{sec:introduction}
\label{sec:guide}

An arbitrary operation on a closed quantum system is represented by a
unitary transformation. Compiling such a transformation into elementary
operations is a fundamental problem in quantum computation. In the
standard circuit model with continuously parameterized one- and
two-qubit gates, the number of gates required to synthesize a generic
unitary scales quadratically with the dimension of the underlying
Hilbert space, matching its number of degrees of freedom up to constant
factors~\cite{ShendeBullockMarkov2006}. However, a more relevant cost measure for a fault-tolerant architecture is the number of non-Clifford gates, which are substantially more expensive than
Clifford gates.  In the standard Clifford+\(T\) model, clean ancillary workspace
can also be traded for a lower \(T\)-count~\cite{LowKliuchnikovSchaeffer2024}. This brings the following question: How many \(T\) gates are required to
approximately implement an arbitrary classically specified unitary?

General unitary and isometry synthesis has been studied under a wide
range of measures, including gate count, circuit depth, and ancillary space 
~\cite{BarencoEtAl1995,MottonenEtAl2004,ItenEtAl2016,
SunTianYangYuanZhang2023,YuanZhang2023}.  The closely related problem of
quantum-state preparation, including its controlled and sparse variants,
has developed in parallel
~\cite{MottonenEtAl2005,PleschBrukner2011,ZhangLiYuan2022,
ZhangYuan2024,MaoTianSun2024,LiLuo2025,LuoLi2026,LOW+26}.
Complementary work addresses diagonal-unitary synthesis and exact or
approximate compilation over Clifford+\(T\)
~\cite{KliuchnikovMaslovMosca2013Approx,GilesSelinger2013,
RossSelinger2016,ZhangWuLi2023,ZhangHuangLi2024,LiEtAl2023Fast}.
In this paper, we focus on minimizing \(T\)-count for a
full, classically specified unitary.

Several recent works study the relevant fault-tolerant landscape.
Low, Kliuchnikov, and Schaeffer developed space--\(T\) tradeoffs for
state preparation and unitary synthesis based on data lookup and
Householder reflections~\cite{LowKliuchnikovSchaeffer2024}.  Gosset, Kothari, and Wu determined the optimal \(T\)-count for arbitrary state
preparation and diagonal-unitary synthesis and proved an \(\Omega(2^n)\)
lower bound at constant accuracy for a worst-case $2^n$-dimensional full unitary, even in
a broader adaptive, diamond-distance model
~\cite{GossetKothariWu2026}.  Rosenthal showed that an arbitrary unitary can be implemented with
\(O(2^{n/2})\) queries to an oracle that prepares its columns.  Using
a related construction in the standard one- and two-qubit gate model,
he also obtained an exact circuit of depth \(\widetilde O(2^{n/2})\)
using \(\widetilde O(4^n)\) ancillas~\cite{Rosenthal2026}.  These
oracle and depth statements do not directly give comparable compiled
\(T\)-counts. The best previous worst-case \(T\)-count upper bound is achieved by Tan's recursive cosine--sine construction~\cite{Tan2025}.  

Throughout this paper, we consider \(n\)-qubit unitary synthesis with approximation error \(0<\eps<1/2\).  All logarithms are base two,
and \(\widetilde O(\cdot)\) suppresses factors polylogarithmic in \(2^n\)
and \(1/\eps\).  Tan's result gives $T$-count
\begin{equation}
 O\!\left(2^{4n/3}(n+\log(1/\eps))^{2/3}+2^n(n+\log(1/\eps))\right).
  \label{eq:intro-tan-bound}
\end{equation}
Recent work by Fang, Heunen, and Wang gives an instance-dependent bound
in terms of phase-optimized Frobenius distance to the Clifford group
~\cite{FangHeunenWang2026}.  At fixed accuracy, that instance-dependent
branch specializes in the worst case to \(\widetilde O(2^{3n/2})\) and
therefore does not by itself improve Tan's generic worst-case exponent.

We call \(U\) \textit{effectively specified} if its entries can be approximated
to any precision with certified error bounds.  The associated
classical pre-processing time is not included in the resource count. Our main result improves the worst-case exponent in the broad-accuracy regime \(n+\log(1/\eps)\le 2^n\).

\begin{theorem}[Unitary synthesis]
\label{thm:main}
For every effectively specified \(U\in\U(2^n)\), there exists a
Clifford+\(T\) circuit that implements \(U\) with error at most
\(\eps\).  If \(n+\log(1/\eps)\le 2^n\), the circuit has \(T\)-count
\begin{equation}
  O\!\left(n2^{5n/4}(n+\log(1/\eps))^{5/8}\right)
  \label{eq:main-T}
\end{equation}
and uses \(O(2^n\sqrt{n+\log(1/\eps)})\) clean ancillas. 
\end{theorem}
 For \(n+\log(1/\eps)>2^n\), Tan's
construction gives \(O(2^n(n+\log(1/\eps)))\) using \(O(n+\log(1/\eps))\) clean ancillas.

In practice, the accuracy $\epsilon$ is at least inversely exponentially $1/\exp(n)$ large, in which case
\(\log(1/\eps)=\mathrm{poly}(n)\) and hence
our bound for \(T\)-count is
\(\widetilde O(2^{5n/4})\), improving upon Tan's
\(\widetilde O(2^{4n/3})\) bound by reducing the coefficient of \(n\) in the exponent from \(4/3\) to \(5/4\).  Together with the known
\(\widetilde\Omega(2^n)\) lower bound \cite{GossetKothariWu2026}, it leaves a multiplicative gap of \(\widetilde O(2^{n/4})\).

Our construction treats the target unitary as a single block-encoded
object rather than decomposing it into a long product of simpler
unitaries as all previous work did \cite{LowKliuchnikovSchaeffer2024,GossetKothariWu2026,Tan2025}.  Choose a power-of-two block size \(b\) and write \(D=d/b\), where $d=2^n$.
Two Boolean phase oracles and Walsh transforms produce a unitary
\(V\) whose norm distribution over all \(b\)-by-\(b\) blocks is flattened:
\begin{equation}
  \max_{I,J}\norm{V_{IJ}}
  =O\!\left(\sqrt{\frac{b}{d}}\log d\right).
  \label{eq:intro-flattening}
\end{equation}
The two Boolean phase oracles simultaneously ensure the required bound for every block, which is crucial
for encoding the full matrix in a single circuit.

We normalize each block, complete it to a unitary dilation
~\cite{Halmos1950,Robinson2018}, and place all \(D^2\) dilations in a single
jointly compiled \(\SELECT\) operator.  Preparing and later uncomputing the block
labels then gives a block encoding of \(V\) with normalization
\begin{equation}
  \rho
  =O\!\left(\sqrt{\frac{d}{b}}\log d\right).
  \label{eq:intro-normalization}
\end{equation}
This normalization grows only as \(\sqrt{d/b}\), while each selected
dilation acts on a space of dimension \(2b\).  Tan's controlled-unitary
compiler therefore implements one use of this block encoding with
\(O(d\sqrt L+b^2L)\) \(T\) gates~\cite{Tan2025}, where \(L:=n+\log(1/\eps)\).

Since the encoded block is \(V/\rho\) for a unitary \(V\), all of its
singular values are equal to \(1/\rho\).  We apply a degree-\(O(\rho)\)
quantum singular value transformation that maps \(1/\rho\) exactly to
one, thereby recovering \(V\) using \(O(\rho)\) applications of the
block encoding~\cite{Gilyen2019}.  The resulting \(T\)-count is
\begin{equation}
  O\!\left(
    \sqrt{\frac{d}{b}}\log d\,
    [d\sqrt L+b^2L]
  \right).
  \label{eq:intro-tradeoff}
\end{equation}
After amplification, the first contribution decreases with \(b\),
whereas the second increases.  Balancing them at
\(b=\Theta(d^{1/2}L^{-1/4})\) gives the exponent \(5/4\).  The Boolean phase oracles and
Walsh transforms are then undone at a lower-order cost.

Randomized Walsh transforms, unitary dilation, block encoding, and
quantum singular value transformation are all well-established
techniques~\cite{TsengEtAl2024,Halmos1950,Robinson2018,
CladerEtAl2022,Gilyen2019}.  The central new
construction here is the block-dilation \(\SELECT\): after simultaneous
spectral-norm flattening, it organizes all normalized blocks of an
arbitrary unitary into a single \(\SELECT\) and recombines them into a
block proportional to the entire target.  This construction
simultaneously controls the block-encoding normalization and the
dimension on which each selected unitary acts.  Tan's controlled-unitary
compiler makes the resulting \(\SELECT\) inexpensive to implement, while
the balance between its compilation cost and the amplification overhead
is what yields the exponent \(5/4\). The reduction in \(T\)-count uses \(O(d\sqrt L)\) clean ancillas.
The flattening signs can be found by randomized classical
preprocessing: for the chosen block size, a random pair succeeds with constant
probability, and candidate pairs can be checked by evaluating the
corresponding block norms with certified numerical bounds.  This
classical preprocessing is not included in the quantum resource count.


The remainder of the paper is organized as follows. Section~\ref{sec:pre} introduces the notation, implementation model, and preliminary results used throughout the paper. Section~\ref{sec:flattening} establishes the simultaneous block-flattening lemma for arbitrary unitary matrices. Section~\ref{sec:block-encoding} then constructs a block encoding of the resulting block-flattened unitary. In Section~\ref{sec:amplification}, we prove the \(T\)-count and ancilla-count bounds for general unitary synthesis. Section~\ref{sec:multiplexed} extends the circuit framework to uniformly controlled unitaries. Section~\ref{sec:conclusion} concludes with a discussion of the remaining gap between the upper and lower bounds.

\section{Preliminaries}
\label{sec:pre}

In this section, we introduce the notation and preliminary lemmas that will be utilized in the rest of the paper.

\paragraph*{Notation.}
For a linear map \(A:\mathcal H_1\to\mathcal H_2\) between
finite-dimensional Hilbert spaces, we denote its operator norm by
\(
  \norm{A}:=\max_{\norm{x}_2=1}\norm{Ax}_2.
\)
 A Rademacher random variable $\zeta$ is a random variable such that $\Pr\{\zeta=1\}=\Pr\{\zeta=-1\}=1/2$. The $n$-qubit Walsh transform is the unitary $H_{d}:= H^{\otimes n}$, where $H$ is the Hadamard gate and $d=2^n$. For an integer $k\ge 1$ and a register $\X$, let $I_k$ denote the $k$-dimensional identity operator and $I_{\X}$ denote the identity operator on register $\X$.

\begin{definition}
\label{def:clean}
For an integer \(m\ge0\), define \(J_{0}:\C^{2^n}\longrightarrow\C^{2^n}\otimes\C^{2^m}\) to append \(m\) ancillas
\begin{equation}
  J_{0}\ket{\psi}=\ket{\psi}\ket{0^m}.
  \label{eq:clean-embedding-definition}
\end{equation}
For \(U\in\U(2^n)\) and \(\eps\ge0\), an \((n+m)\)-qubit unitary circuit \(\widetilde U\in\U(2^{n+m})\) \textit{implements \(U\) with error \(\eps\)} if
\begin{equation}
  \norm{\widetilde UJ_{0}-J_{0}U}\le\eps.
  \label{eq:clean-error}
\end{equation}
\end{definition}

We refer to this as an \(\eps\)-approximate implementation of \(U\)
using \(m\) clean ancillas. This is precisely the approximation
criterion used by Tan~\cite[Definition~1.2]{Tan2025}, expressed using
the embedding \(J_0\). It controls the action of \(\widetilde U\)
uniformly on the clean-ancilla input subspace, including leakage from
that subspace, but places no restriction on its orthogonal complement.

In this paper, we study general unitary synthesis in the Clifford+\(T\) circuit model. A Clifford+\(T\) circuit consists of Clifford and \(T\) gates, where the Clifford group is generated by the Hadamard, phase, and CNOT gates. The \(T\)-count of a circuit is the total number of \(T\) gates it contains.

\begin{problem}[\(T\)-count optimization for unitary synthesis]\label{pro:us}
Let $U\in \U(2^{n})$ and $\eps\in (0,1)$.
The \(T\)-count optimization problem for unitary synthesis is to find an
\((n+m)\)-qubit Clifford+\(T\) circuit
\(\widetilde U\in\U(2^{n+m})\), for some \(m\ge0\), that
\(\eps\)-approximately implements \(U\) while minimizing the
\(T\)-count.
\end{problem}

The following lemmas will be utilized in our subsequent circuit construction for general unitary synthesis.
\begin{lemma}[Uniformly controlled unitary (UCU),  \cite{Tan2025}]
\label{lem:tan-controlled}
Let \(1\le k<n\). We call $R$ an $(n-k,k)$-uniformly controlled unitary (UCU) if
\begin{equation}
  R=\sum_{x\in\{0,1\}^{n-k}}
      \proj{x}\otimes R_x,
  \qquad R_x\in\U(2^k).
  \label{eq:generic-controlled-unitary}
\end{equation}  
For every \(\eps\in(0,1)\), there exists a Clifford+\(T\) circuit
\(\widetilde R\) that \(\eps\)-approximately implements \(R\).
Its \(T\)-count and ancilla-count are both
\begin{equation}
  O\!\left(
    2^{(n+k)/2}\sqrt{k+\log(1/\eps)}
    +4^k\bigl(k+\log(1/\eps)\bigr)
  \right).
  \label{eq:tan-controlled-cost}
\end{equation}
\end{lemma}

\begin{lemma}[Boolean phase oracle, \cite{GossetKothariWu2026}]
\label{lem:boolean-phase}
An $n$-qubit Boolean phase oracle $D\in\U(2^n)$ is a diagonal unitary whose diagonal entries are all $\pm 1$.
Every such unitary can be implemented exactly by a Clifford+\(T\)
circuit using \(O(2^{n/2})\) \(T\) gates and \(O(2^{n/2})\) clean
ancillas.
\end{lemma}

\begin{lemma}[Rectangular matrix Rademacher series, \cite{Tropp2012}]\label{lem:rademacher}
    Consider a finite sequence $\{B_\ell\}$ of fixed matrices in $\C^{d_1\times d_2}$, and let $\{\gamma_\ell\}$ be a finite sequence of independent Rademacher random variables. Define the variance parameter
    \begin{equation}
        \sigma^2:= \max\left\{\norm{\sum_{\ell} B_\ell B_\ell^\dagger},\norm{\sum_{\ell} B^\dagger_\ell B_\ell}\right\}.
    \end{equation}
    Then, for every \(t\ge0\),
    \begin{equation}
        \Pr\left\{\norm{\sum_{\ell}\gamma_\ell B_\ell}\ge t\right\}\le (d_1+d_2)\cdot \exp\left(-\frac{t^2}{2\sigma^2}\right).
    \end{equation}
\end{lemma}

\begin{lemma}[\(R_z\) gate, \cite{RossSelinger2016}]\label{lem:Rz-Tcount}
    For any $\theta\in\mathbb{R}$ and \(0<\eps<1\), the \(R_z\) gate $R_z(\theta)$ can be implemented up to error $\eps$ using $O(\log(1/\eps))$ $T$ gates.
\end{lemma}

\begin{lemma}[Toffoli gate, \cite{Maslov2016}]
    An $n$-qubit generalized Toffoli gate can be implemented exactly using $O(n)$ $T$ gates and ancillas.
\end{lemma}

\section{Simultaneous block flattening of unitary matrices}
\label{sec:flattening}
This section shows that any unitary $U\in \U(2^n)$ can be transformed into $V\in\U(2^n)$ by two Walsh transforms and two Boolean phase oracles, so that all of $V$'s blocks have uniformly bounded operator norms.

Fix an integer \(k\) with \(1\le k\le n\), and set \(r=n-k\). Let
\begin{equation}\label{eq:parameter}
  d=2^n,\quad b=2^k \quad \text{and} \quad D=d/b=2^{r}.
\end{equation}
We split the computational basis into \(D\) consecutive blocks, each
of dimension \(b\).  Equivalently, identify
\(\C^d\cong\C^D\otimes\C^b\) and write its basis states as
\(\ket{J,y}\), where \(0\le J<D\) is the block address and
\(0\le y<b\) is the position inside that block.  For each \(0\le I<D\), define \(P_I:\C^d\longrightarrow\C^b\) on basis states by
\begin{equation}
  P_I\ket{J,y}=\delta_{IJ}\ket y.
  \label{eq:block-selector}
\end{equation}
Thus \(P_I\) selects block \(I\), while \(P_I^\dagger\) embeds a
\(b\)-vector into that block.  These maps obey
\(
  P_IP_J^\dagger=\delta_{IJ}I_b,\) and \(
  \sum_{I=0}^{D-1}P_I^\dagger P_I=I_d.
\)
The \((I,J)\)-block of a matrix \(V\in\C^{d \times d}\) is
\(P_IVP_J^\dagger\).

\begin{lemma}
\label{lem:flatten}
For every \(U\in\U(d)\), there exist Boolean phase oracles $D_1,D_2\in \U(d)$ such that every block of
\begin{equation}
  V=H_dD_1UD_2H_d
  \label{eq:V-definition}
\end{equation}
satisfies
\begin{equation}
  \norm{P_IVP_J^\dagger}
  \le
  g:=\min\!\left\{
    1,
    C_{\mathrm{flat}}\sqrt{b/d}\log(2d)
  \right\},
  \label{eq:block-flat}
\end{equation}
where $0\le I,J<D$ and $C_{\mathrm{flat}}=16\ln2$.

\end{lemma}

\begin{proof}
Choose \(D_2=\diag(\xi_1,\ldots,\xi_d)\) where the $\xi_\ell$ are independent Rademacher random variables.  For a fixed row \(p\) and column \(q\),
\begin{equation}
  (UD_2H_d)_{pq}
  =\frac1{\sqrt d}\sum_{\ell=1}^d
     U_{p\ell}\xi_\ell\sigma_{\ell q},
  \quad
  \sigma_{\ell q}\in\{+1,-1\}.
  \label{eq:entry-rademacher-sum}
\end{equation}
If \(c_\ell=d^{-1/2}U_{p\ell}\sigma_{\ell q}\), then unitarity of
\(U\) gives
\begin{equation}
  \sum_{\ell=1}^d|c_\ell|^2
  =\frac1d\sum_{\ell=1}^d|U_{p\ell}|^2
  =\frac1d.
  \label{eq:entry-coefficient-norm}
\end{equation}
Let $u_\ell=\operatorname{Re} c_\ell$ and $v_\ell=\operatorname{Im} c_\ell$. We can verify that 
\begin{align}
    &\mathbb{E}[{\operatorname{Re}(UD_2H_d)_{pq}}]=\mathbb{E}[\sum_{\ell=1}^d \xi_\ell u_\ell]=\sum_{\ell=1}^d \mathbb{E}[\xi_\ell] u_\ell=0,\\
    &\mathbb{E}[{\operatorname{Im}(UD_2H_d)_{pq}}]=\mathbb{E}[\sum_{\ell=1}^d \xi_\ell v_\ell]=\sum_{\ell=1}^d \mathbb{E}[\xi_\ell] v_\ell=0.
\end{align}
Combined with Hoeffding's inequality, this implies
\begin{align}
    &\Pr_{D_2}\{|\operatorname{Re} (UD_2 H_d)_{pq}|\ge t \}\le 2\exp\left(-\frac{dt^2}{2}\right),\\
    &\Pr_{D_2}\{|\operatorname{Im} (UD_2 H_d)_{pq}|\ge t \} \le 2\exp\left(-\frac{dt^2}{2}\right),
\end{align}
for all $t\ge 0$.
If \(|(UD_2H_d)_{pq}| \ge \alpha\), then at least one of
\(|\operatorname{Re} (UD_2H_d)_{pq}|\) and \(|\operatorname{Im} (UD_2H_d)_{pq}|\) is at least
\(\alpha/\sqrt2\).  Therefore
\begin{multline}
  \Pr\!\left\{|(UD_2H_d)_{pq}|\ge \alpha\right\}
  \le   \Pr\!\left\{|\operatorname{Re}(UD_2H_d)_{pq}|\ge \frac{\alpha}{\sqrt{2}}\right\}\\+\Pr\!\left\{|\operatorname{Im}(UD_2H_d)_{pq}|\ge \frac{\alpha}{\sqrt{2}}\right\}
  \le 4\exp\!\left(-\frac{d\alpha^2}{4}\right).
  \label{eq:complex-entry-tail}
\end{multline}
Set $\alpha=4\sqrt{\frac{\ln(2d)}d}$.
There are \(d^2\) entries, so the union bound gives
\begin{align}
  &\Pr\!\left\{
     \max_{p,q}|(UD_2H_d)_{pq}|\ge\alpha
  \right\}\\
   &\le 4d^2\exp\!\left(-\frac{d\alpha^2}{4}\right)=4d^2(2d)^{-4}
    =\frac1{4d^2}
    \le\frac14.
  \label{eq:entry-union-bound}
\end{align}
Hence there exists a choice of \(D_2\) satisfying
\begin{equation}
  \max_{p,q}|(UD_2H_d)_{pq}|<\alpha.
  \label{eq:entry-flat}
\end{equation}
Fix such a choice for the rest of the proof.
   
For each input column block \(J\), define the \(d\times b\) column slab
\begin{equation}
  Z_J=UD_2H_dP_J^\dagger.
  \label{eq:ZJ-definition}
\end{equation}
Its columns are orthonormal:
\begin{equation}
  Z_J^\dagger Z_J=I_b.
  \label{eq:ZJ-isometry}
\end{equation}
Write its \(\ell\)-th row as \(z_\ell^\dagger\), where
\(z_\ell\in\C^b\).  Every row contains \(b\) entries covered by
Eq.~\eqref{eq:entry-flat}, and therefore
\begin{equation}
  \norm{z_\ell}^2
  <b\alpha^2
  =16\frac bd\ln(2d)
  =:\mu.
  \label{eq:row-coherence}
\end{equation}

Now choose \(D_1=\diag(\eta_1,\ldots,\eta_d)\) with independent Rademacher random
variables.  Fix \(I,J\), let \(e_\ell\) be the \(\ell\)-th standard basis
vector, and put $a_\ell=P_IH_de_\ell\in\C^b$. 
\begin{equation}
    \norm{a_\ell}^2=\norm{P_I H_d e_\ell}^2=b/d.\label{eq:a-ell-norm}
\end{equation}
Since \(D_1=\sum_\ell\eta_\ell e_\ell e_\ell^\dagger\), we have
\begin{align}
  P_IH_dD_1Z_J
   =\sum_{\ell=1}^d
      \eta_\ell P_IH_de_\ell e_\ell^\dagger Z_J =\sum_{\ell=1}^d\eta_\ell a_\ell z_\ell^\dagger.
  \label{eq:matrix-rademacher-sum}
\end{align}
This is a matrix Rademacher series with coefficient matrices
\(B_\ell=a_\ell z_\ell^\dagger\in\C^{b\times b}\).

We can verify that
\begin{align}
   \sum_{\ell=1}^d a_\ell a_\ell^\dagger
   &=P_IH_d\left(\sum_{\ell=1}^de_\ell e_\ell^\dagger\right)
       H_d^\dagger P_I^\dagger
     =I_b.
  \label{eq:a-resolution}
\end{align}
Together with Eqs.~\eqref{eq:ZJ-isometry}, \eqref{eq:row-coherence} and \eqref{eq:a-ell-norm}, the two rectangular variance matrices satisfy
\begin{align}
  \sum_{\ell=1}^dB_\ell B_\ell^\dagger
    &=\sum_{\ell=1}^d
       \norm{z_\ell}^2a_\ell a_\ell^\dagger \preceq
       \mu\sum_{\ell=1}^d a_\ell a_\ell^\dagger
     =\mu I_b, \label{eq:left-variance}\\
     \sum_{\ell=1}^dB_\ell^\dagger B_\ell
   &=\sum_{\ell=1}^d
      \norm{a_\ell}^2z_\ell z_\ell^\dagger =\frac bd\sum_{\ell=1}^d z_\ell z_\ell^\dagger
    =\frac bd Z_J^\dagger Z_J
    =\frac bd I_b.
  \label{eq:right-variance}
\end{align}
The variance parameter is thus at most
\begin{align}
  \sigma^2
   =&\max\!\left\{
      \norm{\sum_\ell B_\ell B_\ell^\dagger},
      \norm{\sum_\ell B_\ell^\dagger B_\ell}
     \right\}\\
     \le & \max\left\{\mu,\frac{b}{d}\right\}
   =\mu,
  \label{eq:matrix-variance}
\end{align}
where \(\mu\ge b/d\) because \(16\ln(2d)>1\).

Applying Lemma~\ref{lem:rademacher} with \(d_1=d_2=b\) and \(t=16\sqrt{b/d}\ln(2d)\) gives
\begin{align}
   &\Pr_{D_1}\left\{\norm{P_I H_d D_1 Z_J}\ge 16\sqrt{\frac{b}{d}} \ln(2d)\right\}\\
  =&\Pr_{D_1}\!\left\{
    \norm{\sum_\ell\eta_\ell B_\ell}\ge 16\sqrt{\frac{b}{d}} \ln(2d)
  \right\}\\
  \le & 2b\exp\!\left(-\frac{16^2(b/d)\ln^2 (2d)}{2\mu}\right)=2b(2d)^{-8}.
  \label{eq:tropp-rectangular-stated}
\end{align}
Note that $P_IH_dD_1Z_J=P_IH_dD_1UD_2H_dP_J^\dagger=P_IVP_J^\dagger$.
A union bound over all \(D^2\) pairs gives
\begin{align}
  &\Pr_{D_1}\!\left\{
    \exists I,J:\norm{P_IH_dD_1Z_J}\ge 16\sqrt{\frac{b}{d}} \ln(2d)
  \right\}\\
  =&
  \Pr_{D_1}\!\left\{\exists I,J:\norm{P_IVP_J^\dagger}\ge 16\sqrt{\frac{b}{d}} \ln(2d)\right\}\\
   \le & 2bD^2(2d)^{-8}\le 2d^2(2d)^{-8}<1/2^7.
  \label{eq:all-block-union-bound}
\end{align}
Therefore, there exist Boolean phase oracles $D_1$ and $D_2$ such that $\norm{P_I V P_J^\dagger}\le 16\sqrt{b/d}\ln(2d)
=(16\ln2)\sqrt{b/d}\log(2d)$ for $V=H_d D_1 U D_2H_d$ and $0\le I,J <D$. Moreover, we have $\norm{P_IVP_J^\dagger}
  \le\norm{P_I}\norm{V}\norm{P_J^\dagger}
  \le1$. This completes the proof.
\end{proof}
\begin{remark}
This section constructs a preconditioned, block-flattened unitary \(V\).  In the subsequent sections we synthesize \(V\); undoing the preconditioning in Lemma~\ref{lem:flatten} then recovers \(U\) at an additional cost of \(O(\sqrt d)\) \(T\) gates, because the Walsh transforms are Clifford.

If \(U\) were block-encoded directly, its blocks would satisfy only the
trivial bound \( \norm{P_IUP_J^\dagger} \le1\). This would result in a larger
block-encoding normalization, thereby increasing the degree---and hence
the number of block-encoding queries---required by the QSVT step, as
well as the overall \(T\)-count. In contrast, block flattening ensures
that \(\norm{P_IVP_J^\dagger}\le g\) for all \(I,J\), where \(g\ll1\) in
the relevant parameter regime. The smaller block norms reduce the
normalization factor and lead to a more \(T\)-efficient construction. 
\end{remark}

\section{Block encoding of a block-flattened unitary}
\label{sec:block-encoding}
In this section, we construct a quantum circuit that implements a block encoding of a block-flattened unitary.

Let $d,b,D$ be as defined in Eq.~\eqref{eq:parameter}. For any unitary $U\in\U(d)$, choose a corresponding block-flattened unitary \(V\in\U(d)\) guaranteed by Lemma~\ref{lem:flatten}. Define
\begin{equation}
  V_{IJ}:=P_IVP_J^\dagger,
  \qquad
  C_{IJ}:=\frac{V_{IJ}}g.
  \label{eq:Cij-definition}
\end{equation}
Lemma~\ref{lem:flatten} shows \(\norm{C_{IJ}}\le1\), so $C_{IJ}$ is a contraction.

For a contraction \(C_{IJ}\in\C^{b\times b}\), define its Julia--Halmos
dilation
\begin{equation}
  \operatorname{Hal}(C_{IJ})=
  \begin{pmatrix}
    C_{IJ} & (I-C_{IJ}C_{IJ}^\dagger)^{1/2}\\[1mm]
    (I-C_{IJ}^\dagger C_{IJ})^{1/2} & -C_{IJ}^\dagger
  \end{pmatrix}.
  \label{eq:halmos-dilation}
\end{equation}

In particular, $\operatorname{Hal}(C_{IJ})$ is a block encoding of $C_{IJ}$, i.e.,

\begin{equation}
  (\bra0 \otimes I_b)\operatorname{Hal}(C_{IJ})(\ket0\otimes I_b)=C_{IJ}.
  \label{eq:halmos-top-left}
\end{equation}

For completeness, we check unitarity.  If
\(C_{IJ}=L\Sigma R^\dagger\) is a singular-value decomposition, then
\begin{align}
  &C_{IJ}(I-C_{IJ}^\dagger C_{IJ})^{1/2}\\
   =&L\Sigma(I-\Sigma^2)^{1/2}R^\dagger =(I-C_{IJ}C_{IJ}^\dagger)^{1/2}C_{IJ}.
  \label{eq:defect-intertwining}
\end{align}
Using this identity, direct block multiplication gives
\begin{equation}
  \operatorname{Hal}(C_{IJ})\operatorname{Hal}(C_{IJ})^\dagger=\operatorname{Hal}(C_{IJ})^\dagger\operatorname{Hal}(C_{IJ})=I_{2b}.
  \label{eq:halmos-unitarity}
\end{equation}

To construct the quantum circuit for a block encoding of $V\in \U(d)$, we introduce four registers: an \(r\)-qubit register \(\B\), an \(r\)-qubit register \(\X\), a one-qubit register \(\f\), and a \(k\)-qubit register \(\Y\).  Registers \(\B\) and \(\Y\) are logical input registers, whereas \(\X\) and \(\f\) are ancillary registers initialized to \(\ket{0^r}\) and \(\ket0\), respectively.

Define a $(2r,k+1)$-UCU $\SELECT$ acting on registers $\X,\B,\f$ and $\Y$,
\begin{equation}
  \SELECT_{\X\B\f\Y}  :=\sum_{I,J=0}^{D-1}
       \proj{I}_\X\otimes\proj{J}_\B
       \otimes\operatorname{Hal}(C_{IJ})_{\f\Y}.
  \label{eq:select-definition}
\end{equation}

Construct a unitary \(W\) as follows (also see Fig. \ref{fig:W-circuit}):

\begin{enumerate}[leftmargin=2.5em,label={\arabic*.}]
\item Apply \(H_D=H^{\otimes r}\) to \(\X\). 

\item Apply \(\SELECT\) to $\X,\B,\f$ and $\Y$.

\item Apply \(H_D=H^{\otimes r}\) to \(\B\).

\item Swap registers \(\X\) and \(\B\) by $\operatorname{SWAP}(\X,\B)$.

\end{enumerate}

\begin{figure}[ht]
\centering
\begin{tikzpicture}[
  x=1.35cm,y=0.72cm,
  gate/.style={draw,minimum width=0.9cm,minimum height=0.55cm,fill=blue!3},
  lab/.style={font=\small},
  wire/.style={thick}
]
  \foreach \yy/\name in {3.0/{\(\X:\ket{0^r}\)},2.0/{\(\B:\ket J\)},1.0/{\(\f:\ket0\)},0/{\(\Y:\ket y\)}} {
    \draw[wire] (0,\yy) -- (4.5,\yy);
    \node[lab,left] at (0,\yy) {\name};
  }
  \node[gate] at (0.5,3.0) {\(H_D\)};
  \node[gate,minimum width=1.6cm,minimum height=2.75cm] at (1.75,1.5) {\(\SELECT\)};
  \node[gate] at (3,2.0) {\(H_D\)};
  \draw (4,3.0) node[cross out,draw,minimum size=4pt] {};
  \draw (4,2.0) node[cross out,draw,minimum size=4pt] {};
  \draw (4,3.0) -- (4,2.0);
  \node[lab,above] at (4,3.25) {\(\operatorname{SWAP}(\X,\B)\)};
\end{tikzpicture}
\caption{A quantum circuit implementing the unitary \(W\).}
\label{fig:W-circuit}
\end{figure}
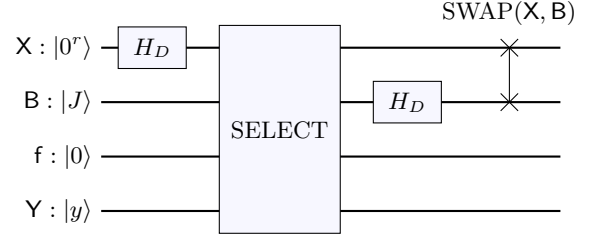

Let the clean embedding $J_0:\mathcal H_{\B\Y}\longrightarrow\mathcal H_{\X\B\f\Y}$ satisfy
\begin{equation}
  J_0(\ket J_\B\ket y_\Y)
   =\ket{0^r}_\X\ket J_\B\ket0_\f\ket y_\Y.
  \label{eq:J0-definition}
\end{equation}

The following lemma shows that $W$ is a block encoding of $V$.
\begin{lemma}
\label{lem:uniform-block}
Let $\rho=Dg$. Then
\begin{equation}
  J_0^\dagger WJ_0
  =\frac1D\sum_{I,J=0}^{D-1}
      \ket I\!\bra J_\B\otimes (C_{IJ})_{\Y}
  =\frac V\rho,
  \label{eq:uniform-clean-block}
\end{equation}

\end{lemma}

\begin{proof}
For $0\le J <D$ and $0\le y < b$, the initial state
is
\begin{equation}
  \ket{0^r}_\X\ket J_\B\ket0_\f\ket y_\Y.
\end{equation}
After the first $H_D$ on $\X$, it is
\begin{equation}
  \frac1{\sqrt D}\sum_{I=0}^{D-1}
    \ket I_\X\ket J_\B\ket0_\f\ket y_\Y.
  \label{eq:W-state-after-first-H}
\end{equation}
Applying \(\SELECT\), we obtain
\begin{align}
  &\frac1{\sqrt D}\sum_{I=0}^{D-1}
    \ket I_\X\ket J_\B \operatorname{Hal}(C_{IJ})(\ket0_\f\ket y_\Y)\\
    &=\frac1{\sqrt D}\sum_{I=0}^{D-1}
    \ket I_\X\ket J_\B \ket0_\f C_{IJ} \ket y_\Y + \ket{\perp},
  \label{eq:W-state-after-select}
\end{align}
where 
\begin{equation}
  \ket{\perp}:=\frac1{\sqrt D}\sum_{I=0}^{D-1}
    \ket I_\X\ket J_\B \ket1_\f (I-C_{IJ}^\dagger C_{IJ})^{1/2}\ket y_\Y  
\end{equation}
and $\ket{\perp}$ is orthogonal to the first component. For every computational basis label \(\ket{J}\), $\bra{0^r} H_D\ket J=1/{\sqrt D}$. After applying the $H_D$ on $\B$, we have
\begin{align}
    &\frac1{\sqrt D}\sum_{I=0}^{D-1}
    \ket I_\X H_D\ket J_\B \ket0_\f C_{IJ} \ket y_\Y +(H_D)_{\B}\ket{\perp}\\
    =&\frac1{D}\sum_{I=0}^{D-1}
    \ket I_\X \ket {0^r}_\B \ket0_\f C_{IJ} \ket y_\Y + \ket{\perp'},
\end{align}
where
\begin{equation}
    \ket{\perp'}:=(H_D)_{\B}\ket{\perp} + \frac{1}{D}\sum_{I=0}^{D-1}\sum_{K=1}^{D-1}\ket{I}_{\X}\sigma_{JK}\ket{K}_{\B}\ket{0}_{\f}C_{IJ}\ket{y}_{\Y}
\end{equation}
where \(J,K\) are identified with their \(r\)-bit strings and \(\sigma_{JK}:=(-1)^{J\cdot K}\) and $\ket{\perp'}$ is orthogonal to the first component.
The final swap maps the first component to
\begin{equation}
  \frac1D\sum_{I=0}^{D-1}
    \ket0_\X\ket I_\B\ket0_\f C_{IJ}\ket y_\Y.
  \label{eq:W-clean-state-after-swap}
\end{equation}
The image of \(\ket{\perp'}\) under the swap has either
\(\X\ne0^r\) or \(\f=1\), and is therefore annihilated by \(J_0^\dagger\).

Applying \(J_0^\dagger\) leaves
\begin{align}
  &J_0^\dagger WJ_0(\ket{J}_{\B}\ket{y}_{\Y})
  =\frac1D\sum_{I=0}^{D-1}
      \ket I_{\B}C_{IJ}\ket y_{\Y}\notag\\
  = &\frac1{Dg}\sum_{I=0}^{D-1}
      \ket I_{\B}V_{IJ}\ket y_{\Y}
   =\frac V{Dg}(\ket J_{\B}\ket y_{\Y}).
  \label{eq:clean-action-on-block-basis}
\end{align}
This completes the proof.
\end{proof}

Let \(\mathcal K_{\mathrm{work}}\) be the compiler's clean workspace and let
\(\iota:\mathcal H_{\X\B\f\Y}\to
\mathcal H_{\X\B\f\Y}\otimes\mathcal K_{\mathrm{work}}\) append its
all-zero state.  Set 
\begin{equation}\label{eq:J}
    J=\iota J_0.
\end{equation}
and define the ideal lifted query
\begin{equation}
    \widehat W:=W\otimes I_{\mathrm{work}};
\end{equation}
then
\(\widehat W\iota=\iota W\).
The following lemma shows the $T$-count of unitary $W$ in Lemma \ref{lem:uniform-block}.
\begin{lemma}
\label{lem:compiled-base}
Assume \(k<n\). For every \(0<\delta<1/2\), there is a Clifford+\(T\) unitary circuit
\(\widetilde W\) such that
\begin{equation}
  \norm{\widetilde W\,\iota-\iota W}\le\delta.
  \label{eq:compiled-W-guarantee}
\end{equation}
Its $T$-count and ancilla-count are both
\begin{equation}
  O\!\left(d\sqrt{k+\log (1/\delta)}+b^2(k+\log (1/\delta))\right).
  \label{eq:compiled-base-cost}
\end{equation}
Moreover, for $A=J^\dagger\widetilde WJ$,
\begin{equation}
  \norm{A-\rho^{-1}V}\le\delta.
  \label{eq:actual-clean-block-bound}
\end{equation}
\end{lemma}
\begin{proof}
As discussed above, $W$ consists of two Walsh transforms $H_D=H^{\otimes r}$, a $(2r,k+1)$-UCU $\SELECT$ and a swap operator $\operatorname{SWAP}(\X,\B)$. Applying Lemma~\ref{lem:tan-controlled} with error \(\delta\),
we obtain a Clifford+\(T\) implementation of \(\SELECT\) using
\(
  O\!\left(
    d\sqrt{k+\log(1/\delta)}
    +b^2\bigl(k+\log(1/\delta)\bigr)
  \right)
\)
\(T\) gates and ancillas.  Since \(H_D\) and
\(\operatorname{SWAP}(\X,\B)\) are Clifford operations, adjoining them does not increase the \(T\)-count or the approximation error.
Thus the resulting circuit \(\widetilde W\) is a Clifford+\(T\)
circuit with the claimed bounds.

Now use \(J=\iota J_0\), \(\iota^\dagger\iota=I\):
\begin{align}
  \norm{A-\rho^{-1}V}
   &=\norm{
       J_0^\dagger\iota^\dagger\widetilde W\iota J_0
       -J_0^\dagger WJ_0
     }\notag\\
   &=\norm{
       J_0^\dagger\iota^\dagger
       (\widetilde W\iota-\iota W)J_0
     }\notag\\
   &\le
     \norm{J_0^\dagger}\norm{\iota^\dagger}
     \norm{\widetilde W\iota-\iota W}\norm{J_0}\notag\\
   &\le\delta.
  \label{eq:actual-clean-block-expanded}
\end{align}
All embeddings are isometries and hence have norm one.
\end{proof}

\section{$T$-count for general unitary synthesis}
\label{sec:amplification}

In this section, we first use Quantum Singular Value Transformation (QSVT) to construct a circuit for the block-flattened unitary and then use this circuit to obtain a Clifford+\(T\) implementation of a general unitary.

Put
\begin{equation}
  c:=\rho^{-1}.
  \label{eq:c-definition}
\end{equation}
In Eq.~\eqref{eq:uniform-clean-block}, since $W$ is a unitary, 
\begin{equation}
    \norm{J_0^\dagger W J_0}=\norm{\frac{V}{\rho}}=\frac{1}{\rho}=c\le 1.
\end{equation}
Moreover, $g\le 1$ implies $\rho=Dg\le D\le d$. Therefore
\begin{equation}
   0 < c \le 1,\qquad Q:=\Theta(\rho)=O(d).
\end{equation}

\subsection{Quantum circuit for the block-flattened unitary}
The next lemma constructs a
bounded odd polynomial that equals one exactly at the specified point
\(c\).  It also explains precisely which QSVT results make that polynomial implementable.

\begin{lemma}
\label{lem:exact-response}
For every \(0<c\le1\), there is an odd integer $Q=\Theta(1/c)$ and a degree-$Q$ polynomial $P(x)\in\mathbb{C}[x]$
such that $P(x)$ satisfies $P(c)=1$ and
\begin{enumerate}
        \item $P$ has parity-$(Q\mod 2)$,
        \item $\forall x\in[-1,1]$: $|P(x)|\le 1$,
        \item $\forall x\in (-\infty,-1]\cup [1,\infty)$: $|P(x)|\ge 1$,
    \end{enumerate}
\end{lemma}
\begin{proof}
    See the proof in Appendix \ref{append:proof-restate}.
\end{proof}


Recall the QSVT theorem of Gily{\'e}n et al.~\cite[Theorem~17]{Gilyen2019}.
Let \(a_{\mathrm{work}}\) be the number of compiler-workspace qubits and
set \(a:=r+1+a_{\mathrm{work}}\), counting \(\X,\f\), and that workspace.
The signal projector is
\begin{equation}\label{eq:projector}
  \Pi=JJ^\dagger,
\end{equation}
whose range is the subspace in which these \(a\) signal ancillas are all
zero.  For the lifted ideal query \(\widehat W=W\otimes I_{\mathrm{work}}\)
defined above,
\begin{equation}
  J^\dagger\widehat WJ=cV,
  \qquad
  \Pi\widehat W\Pi=J(cV)J^\dagger.
\end{equation}
For the degree-\(Q\) odd polynomial \(P\) from
Lemma~\ref{lem:exact-response}, the QSVT theorem supplies
\(\Phi=(\phi_1,\ldots,\phi_Q)\in\mathbb R^Q\) such that
\[
  P^{(\mathrm{SV})}\!\left(J(cV)J^\dagger\right)
  =P^{(\mathrm{SV})}(\Pi\widehat W\Pi)
  =\Pi U_\Phi\Pi
  =JVJ^\dagger,
\]
where, with the noncommuting product ordered in increasing \(j\),
\begin{equation}\label{eq:QSVT-circuit}
  U_\Phi=e^{\ii\phi_1(2\Pi-I)}\widehat W
  \prod_{j=1}^{(Q-1)/2}
  e^{\ii\phi_{2j}(2\Pi-I)}\widehat W^\dagger
  e^{\ii\phi_{2j+1}(2\Pi-I)}\widehat W.
\end{equation}
For any $j$, $e^{\ii \phi_i(2\Pi-I)}$ is called the signal rotations.
Thus
\begin{equation}
    P^{(\mathrm{SV})}(cV)=V,
\end{equation}
Since \(U_\Phi\) is unitary and its
clean compression is the unitary \(V\), no amplitude leaks from the clean
input subspace; hence 
\begin{equation}
    U_\Phi J=JV.
\end{equation}
Then $U_{\Phi}$ can implement unitary $V$ exactly.

In the next lemma, \(\mathcal W\) is any \emph{exact} unitary, later instantiated as the exact Clifford+T circuit \(\widetilde W\).  The only approximation
hypothesis concerns its clean block \(A=J^\dagger\mathcal WJ\).

\begin{lemma}
\label{lem:robust-amplification}
Let \(\mathcal W\) be a unitary, let \(J\) be an isometry, and put
\begin{equation}
  A=J^\dagger\mathcal WJ.
  \label{eq:robust-A-definition}
\end{equation}
Suppose \(V\) is unitary and
\begin{equation}
  \norm{A-cV}\le\delta,
  \qquad
  0<c\le1,
  \qquad
  \delta<c/2.
  \label{eq:robust-hypothesis}
\end{equation}
There is a unitary \(\mathcal A\), using
\(Q=\Theta(1/c)\) total calls, alternating between \(\mathcal W\) and
\(\mathcal W^\dagger\), and \(O(Q)\) signal rotations, whose clean top block
\begin{equation}
  A_{\mathrm{amp}}:=J^\dagger\mathcal AJ
  \label{eq:Aamp-definition}
\end{equation}
satisfies
\begin{equation}
  \norm{A_{\mathrm{amp}}-V}=O(Q^2\delta).
  \label{eq:amplified-top-error}
\end{equation}
If \(Q^2\delta=O(1)\), then the full
clean input-isometry error is
\begin{equation}
  \norm{\mathcal AJ-JV}=O(Q\sqrt\delta).
  \label{eq:amplified-isometry-error}
\end{equation}
In particular, if \(\delta=0\), \(\mathcal A\) implements \(V\) exactly
and all signal ancillas return to zero.
\end{lemma}

\begin{proof}
Let \(P\) be the degree-\(Q\) response from
Lemma~\ref{lem:exact-response}.  
Circuit $\mathcal{A}$ is obtained by replacing $\widehat W$ by $\mathcal{ W}$ in Eq.~\eqref{eq:QSVT-circuit}, using the same clean projector \(\Pi=JJ^\dagger\). 

Write a singular-value decomposition and polar decomposition of \(A\)
as
\begin{equation}
  A=U_A\Sigma R_A^\dagger=\Omega|A|,
  \label{eq:A-polar-SVD}
\end{equation}
where
\begin{equation}
    \Omega=U_AR_A^\dagger \qquad \text{and} \qquad
  |A|=R_A\Sigma R_A^\dagger.
\end{equation}
Every singular value of \(cV\) is \(c\). Therefore
Eq.~\eqref{eq:robust-hypothesis} gives the following bounds on the singular values of \(A\)
\begin{equation}
  |\sigma_j(A)-c|\le\delta
  \label{eq:singular-values-near-c}
\end{equation}
for every \(j\).  Since \(\delta<c/2\), \(A\) is invertible and
\(\Omega\) is unitary.  Also, because \(cI\) is scalar,
\begin{equation}
  \norm{|A|-cI}
  =\max_j|\sigma_j(A)-c|
  \le\delta.
  \label{eq:absolute-perturbation}
\end{equation}

For an odd response polynomial $P$, QSVT gives
\begin{equation}
  A_{\mathrm{amp}}=P^{(\mathrm{SV})}(A)
   =U_AP(\Sigma)R_A^\dagger
   =\Omega P(|A|).
  \label{eq:qsvt-polar-form}
\end{equation}

We next bound the distance between \(A_{\mathrm{amp}}\) and \(V\).
{ First, \(A=\Omega|A|\) implies the exact identity
\begin{equation}
  c(I-\Omega^\dagger V)
  =(cI-|A|)+\Omega^\dagger(A-cV).
  \label{eq:polar-identity}
\end{equation}
Both terms on the right have norm at most \(\delta\), so
\begin{equation}
  \norm{\Omega-V}
  =\norm{I-\Omega^\dagger V}
  \le\frac{2\delta}{c}.
  \label{eq:polar-perturbation}
\end{equation}

Second, fix \(x\in[-1,1]\), choose
\(\theta_x=\arg P'(x)\), and set
\(q_x(t):=\operatorname{Re}(e^{-\ii\theta_x}P(t))\).  Then \(q_x\) is
a real polynomial bounded by one on \([-1,1]\), so the real Markov
inequality gives \(\max_{x\in[-1,1]}|P'(x)|=|q_x'(x)|\le Q^2\).  Moreover, because
\(A\) is a compression of a unitary, \(0\le\sigma_j(A)\le1\); hence
the integration interval below lies in \([-1,1]\).
Consequently,
{\begin{align}
  |P(\sigma_j(A))-P(c)|
  &=
  \left|\int_c^{\sigma_j(A)}P'(x)\,dx\right| \notag\\
  &\le
  \max_{t\in[-1,1]}|P'(t)|
  |\sigma_j(A)-c| \notag\\
  &\le Q^2\delta .
\end{align}}
Using \(P(c)=1\), Eq.~\eqref{eq:singular-values-near-c}, and the preceding derivative bound,
\begin{align}
  \norm{P(|A|)-I}
   &=\max_j|P(\sigma_j(A))-P(c)|
    \le Q^2\delta.
  \label{eq:polynomial-perturbation}
\end{align}}
Combining Eqs.~\eqref{eq:qsvt-polar-form},
\eqref{eq:polar-perturbation}, and
\eqref{eq:polynomial-perturbation},
\begin{align}
  \norm{A_{\mathrm{amp}}-V}
   &\le\norm{\Omega(P(|A|)-I)}+\norm{\Omega-V}\notag\\
   &\le Q^2\delta+\frac{2\delta}{c}
    \le3Q^2\delta,
  \label{eq:top-error-expanded}
\end{align}
where the last inequality uses \(Q\ge1/c\) from
Eq.~\eqref{eq:Q-explicit-bounds}.  This proves
Eq.~\eqref{eq:amplified-top-error}.

It remains to control ancilla leakage.  Define
\begin{equation}
  E:=(I-JJ^\dagger)\mathcal AJ.
  \label{eq:leakage-operator}
\end{equation}
Then
\begin{equation}
  \mathcal AJ=JA_{\mathrm{amp}}+E,
  \label{eq:A-decomposition-clean-leakage}
\end{equation}
and the two terms on the right have orthogonal ranges.  Unitarity of
\(\mathcal A\) gives
\begin{align}
  E^\dagger E
   &=J^\dagger\mathcal A^\dagger(I-JJ^\dagger)\mathcal AJ\notag\\
   &=I-A_{\mathrm{amp}}^\dagger A_{\mathrm{amp}}.
  \label{eq:leakage-identity}
\end{align}
Put \(e=\norm{A_{\mathrm{amp}}-V}\).  Since \(A_{\mathrm{amp}}\) is a
compression of a unitary, \(\norm{A_{\mathrm{amp}}}\le1\).  Using
\(V^\dagger V=I\),
\begin{align}
  \norm{E}^2
   &=\norm{I-A_{\mathrm{amp}}^\dagger A_{\mathrm{amp}}}\notag\\
   &=\norm{V^\dagger(V-A_{\mathrm{amp}})
      +(V^\dagger-A_{\mathrm{amp}}^\dagger)A_{\mathrm{amp}}}\notag\\
   &\le e+e\norm{A_{\mathrm{amp}}}
    \le2e.
  \label{eq:leakage-from-top-error}
\end{align}
Finally,
\begin{align}
  \norm{\mathcal AJ-JV}
   &\le\norm{J(A_{\mathrm{amp}}-V)}+\norm{E}\notag\\
   &\le e+\sqrt{2e}
    =O(Q\sqrt\delta)
  \label{eq:clean-error-final-robust}
\end{align}
when \(Q^2\delta=O(1)\).  If \(\delta=0\), then \(e=0\) and
\(E=0\), proving exact logical action and exact ancilla return.
\end{proof}

\subsection{Clifford+$T$ circuit for the general unitary}
In this section, we show the $T$-count for the general unitary.

First we implement the signal rotation
\begin{equation}
 e^{\ii\phi(2\Pi-I)}
  \label{eq:centered-signal-phase}
\end{equation}
in the QSVT circuit Eq.~\eqref{eq:QSVT-circuit}.

\begin{lemma}
\label{lem:signal-phase}
Suppose an alternating sequence contains \(M\) known signal rotations
\(e^{\ii\phi_1(2\Pi-I)},\ldots,e^{\ii\phi_M(2\Pi-I)}\).  They can be implemented total error at most $\eta$ using
\begin{equation}
  O\!\left(M[a+\log (M/\eta)]\right)
  \label{eq:signal-phase-cost}
\end{equation}
T gates and $O(a)$ ancillas.
\end{lemma}
\begin{proof}
    Recall that $\Pi=JJ^\dagger=\ket{0^a}\bra{0^a}$ (Eq.~\eqref{eq:projector}). Ref.~\cite{Gilyen2019} shows that $e^{\ii\phi_M(2\Pi-I)}$ can be decomposed into two $O(a)$-qubit Toffoli gates and one Rz gate. The Toffoli gate can be implemented exactly using $O(a)$ $T$ gates and ancillas. And each Rz gate can be approximated using $O(\log(M/\eta))$ $T$ gates up to error $\eta/M$ by Lemma \ref{lem:Rz-Tcount}. Then the total error is $\eta$ and the total $T$-count is $O\!\left(M[a+\log (M/\eta)]\right)$.
\end{proof}



Now we show the $T$-count and ancilla-count for the general unitary.
\begin{theorem}[Restatement of Theorem \ref{thm:main}]
\label{thm:main-restate}
Let \(d=2^n\), let \(0<\eps<1/2\), and set \(L:=n+\log(1/\eps)\).  For any \(U\in\U(d)\), there exists a Clifford+\(T\) circuit that implements \(U\) with error at most \(\eps\) in the sense of Definition~\ref{def:clean}.  If \(L\le d\), the circuit has \(T\)-count
\begin{equation}
  O\!\left(d^{5/4}L^{5/8}\log d\right)
\end{equation}
 and uses \(O(d\sqrt L)\) clean ancillas.  
\end{theorem}
\begin{proof}
  For any $U\in\U(d)$, by Lemma \ref{lem:flatten}, there exist two Boolean phase oracles $D_1,D_2$ such that
\begin{equation}
  U=D_1H_d\,V\,H_dD_2,
  \label{eq:undo-flattening}
\end{equation}
and $V$ satisfies Eq.~\eqref{eq:block-flat}.
The Walsh transforms are Clifford operations.  By
Lemma~\ref{lem:boolean-phase}, the two Boolean phase oracles can be implemented exactly using \(O(\sqrt d)\) \(T\) gates and \(O(\sqrt d)\) clean ancillas in total.  We next construct the circuit for \(V\).

We now apply Lemma~\ref{lem:robust-amplification} to the exact compiled unitary \(\mathcal W=\widetilde W\).  
By Eq.~\eqref{eq:QSVT-circuit} and Lemma~\ref{lem:robust-amplification}, the QSVT circuit for \(V\) contains \(O(Q)=O(\rho)\) calls in total to \(\widetilde W\) and \(\widetilde W^\dagger\), together with \(O(Q)\) signal rotations \(e^{\ii\phi_j(2\Pi-I)}\).

Lemma~\ref{lem:compiled-base} shows
\begin{equation}
  \norm{J^\dagger\widetilde WJ-cV}\le\delta.
  \label{eq:compiled-block-ready-for-amp}
\end{equation}
For any $\eps\in(0,1)$, choose
\begin{equation}
  \delta=\frac{\eps^2}{C\cdot Q^2}
  =\Theta\!\left(\frac{\eps^2}{\rho^2}\right)
  \label{eq:delta-choice}
\end{equation}
for some sufficiently large constant $C>0$. Then $Q^2\delta=\eps^2/C$ is at most a sufficiently small constant and $\delta=\eps^2/(CQ^2)\le \eps^2c^2/C<c/2$, where the first inequality uses \(Q\ge1/c\); this satisfies the condition in Lemma \ref{lem:robust-amplification}.
For sufficiently large \(C\), Lemma~\ref{lem:robust-amplification} implies that the QSVT circuit built from the unitary \(\widetilde W\) and its inverse has error at most \(\eps/2\).
Using \(L:=n+\log(1/\eps)\) as defined in the theorem, set
\begin{equation}
  S_0:= d\sqrt L+b^2L.
  \label{eq:S0-definition}
\end{equation}
Lemma~\ref{lem:compiled-base} and Eq.~\eqref{eq:delta-choice} show that one
query to \(\widetilde W\) or \(\widetilde W^\dagger\) costs 
 {\begin{align}
    &O\left(d\sqrt{k+\log(1/\delta)}+b^2(k+\log(1/\delta))\right)\notag\\
    =&O\left(d\sqrt{k+\log(1/\eps)+\log(Q)}+b^2(k+\log(1/\eps)+\log(Q))\right)\notag\\
    =&O(d\sqrt{L}+b^2L)\notag\\
    =&O(S_0)
\end{align}}
$T$ gates and ancillas.  Across all \(O(Q)\) queries the \(T\)-count and number of reusable clean ancillas are, respectively,
\begin{equation}\label{eq:W}
    O(\rho S_0) \qquad \text{and}\qquad O(S_0).
\end{equation}

Each signal predicate checks \(\X,\f\), and the compiler-workspace qubits.  Thus \(a=r+1+a_{\mathrm{work}}=O(S_0)\). The circuit in Eq.~\eqref{eq:QSVT-circuit} contains \(M=O(Q)\) signal rotations.  Moreover,
\(Q=O(d)\) and \(\eta=\Theta(\eps)\), so
\begin{equation}
  \log (M/\eta)
  =O(n+ {\log(1/\eps)})=O(L).
  \label{eq:signal-log-is-L}
\end{equation}
 {The QSVT circuit has clean-input error at most \(\eps/2\) by Lemma~\ref{lem:robust-amplification}. Set \(\eta=\eps/2\). By Lemma~\ref{lem:signal-phase}, approximating the \(M\) signal rotations introduces error at most
\(\eps/2\).}
Lemma~\ref{lem:signal-phase} shows that when $L\le d$, the total $T$-count and ancilla-count for all signal rotations are
\begin{equation}
  O(\rho (S_0+L))=O(\rho S_0),
  \qquad \text{and} \qquad
  O(S_0).
  \label{eq:amplification-resources}
\end{equation}
Under the assumed regime \(L\le d\), Lemma~\ref{lem:flatten} implies
\begin{equation}
  \rho=Dg
  \le C_{\mathrm{flat}}\sqrt{\frac{d}{b}}\log(2d)
  =O\!\left(\sqrt{\frac{d}{b}}\log d\right).
  \label{eq:rho-resource-bound}
\end{equation}
Combining Eqs.~\eqref{eq:S0-definition},
\eqref{eq:W} and \eqref{eq:amplification-resources}, and adding the \(T\)-count of the two Boolean phase oracles gives the \(T\)-count for \(U\):
\begin{align}
  &O\!\left(\rho[d\sqrt L+b^2L]+\sqrt d\right)\notag\\
  =&O\!\left(
      \log d\left[
        d^{3/2}b^{-1/2}\sqrt L
        +d^{1/2}b^{3/2}L
      \right]
    \right).
  \label{eq:T-before-optimization}
\end{align}
The ancilla-count is
\begin{equation}
    O(S_0)=O(d\sqrt{L}+b^2L)
\end{equation}
Choose $b=\Theta(b_\star)$
\begin{equation}
  b_\star=d^{1/2}L^{-1/4}.
  \label{eq:b-star}
\end{equation}
Such a choice exists when \(L\le d\).
For \(L\le d\), the $T$-count and ancilla-count for general unitary $U$ are
\begin{equation}
  O\!\left(d^{5/4}L^{5/8}\log d\right)\qquad\text{and}\qquad O(d\sqrt{L}).
  \label{eq:T-optimized}
\end{equation}
 This completes the proof.
\end{proof}
When \(L>d\), we instead use the circuit construction in
\cite[Theorem~1.1]{Tan2025}.  Its \(T\)-count satisfies
\begin{equation}
  O\!\left(d^{4/3}L^{2/3}+dL\right)=O(dL),
  \label{eq:tan-high-precision-count}
\end{equation}
and its ancilla count is
\(O(d^{2/3}L^{1/3}+L)=O(L)\).

\section{Synthesis of uniformly controlled unitaries}
\label{sec:multiplexed}

This section applies the block-flattening construction to uniformly controlled unitaries (UCUs).  The additional ingredient is
that a single pair of Boolean phase oracles can flatten every member of the family simultaneously.

Let  
\begin{equation}
    M=2^m\ge2,\quad K=2^k\ge2, \quad 0<\eps<1/2.
\end{equation}
For a classically specified family \(\{U_x\}_{x=0}^{M-1}\), consider the $(m,k)$-UCU
\begin{equation}
  \mathcal U
  :=\sum_{x=0}^{M-1}\proj{x}\otimes U_x,
  \qquad U_x\in\U(K).
  \label{eq:mux-definition}
\end{equation}
Throughout this section, set
\begin{equation}
  R:=\log(2MK),
  \qquad
  \Lambda:=1+\log K+\log(1/\eps).
  \label{eq:mux-parameters}
\end{equation}
Lemma~\ref{lem:tan-controlled} shows that 
\(\mathcal U\) has $T$-count
\begin{equation}
  O\!\left(
    \sqrt M\,K\sqrt\Lambda+K^2\Lambda
  \right).
  \label{eq:direct-mux-Tan}
\end{equation}
We next reduce the quadratic target-side term by applying the block
construction jointly to the entire family.

\subsection{Simultaneous flattening}

Choose a power of two \(h\mid K\), and set \(D=K/h\).  Partition each
\(U_x\) into \(D\times D\) blocks of size \(h\times h\).  Let
\(Q_I:\C^K\to\C^h\) denote the corresponding block selector,
\begin{equation}
  Q_I\ket{J,y}=\delta_{IJ}\ket y,
  \qquad 0\le I,J<D.
  \label{eq:mux-block-selector}
\end{equation}

\begin{lemma}
\label{lem:mux-common-flatten}
There exist Boolean phase oracles \(S_1,S_2\in\U(K)\), independent of
\(x\), such that the unitary
\begin{equation}
  V_x=H_KS_1U_xS_2H_K
  \label{eq:mux-flattened-branch}
\end{equation}
satisfies
\begin{equation}
  \norm{Q_IV_xQ_J^\dagger}
  \le
  g_{\mathrm{mux}}
  :=\min\!\left\{
    1,
    C_{\mathrm{flat}}\sqrt{\frac hK}\,R
  \right\}
  \label{eq:mux-g}
\end{equation}
for every \(x,I,J\), where \(C_{\mathrm{flat}}=16\ln2\).
\end{lemma}

\begin{proof}
Choose \(S_2=\diag(\xi_1,\ldots,\xi_K)\), where the
\(\xi_\ell\) are independent Rademacher random variables.  For fixed
\(x,p,q\), the entry \((U_xS_2H_K)_{pq}\) is a complex Rademacher sum
whose squared coefficient norm is \(1/K\).  Set
\(\alpha=4\sqrt{\ln(2MK)/K}\).  Applying
Eq.~\eqref{eq:complex-entry-tail} and taking a union bound over all
\(MK^2\) entries gives
\begin{align}
  &\Pr\!\left\{
    \max_{x,p,q}|(U_xS_2H_K)_{pq}|\ge\alpha
  \right\}\le4MK^2(2MK)^{-4}<1.
  \label{eq:mux-entry-union-bound}
\end{align}
Fix a choice of \(S_2\) for which this event does not occur.

For each \(x,J\), define the \(K\times h\) isometry
\begin{equation}
  Z_{x,J}=U_xS_2H_KQ_J^\dagger,
  \label{eq:mux-Z}
\end{equation}
and write its \(\ell\)-th row as \(z_{x,J,\ell}^\dagger\).  The entry
bound above implies
\begin{equation}
  \norm{z_{x,J,\ell}}^2
  \le16\frac hK\ln(2MK)
  =:\mu_{\mathrm{mux}}.
  \label{eq:mux-row-coherence}
\end{equation}

Now choose \(S_1=\diag(\eta_1,\ldots,\eta_K)\), where \(\eta_1,\ldots,\eta_K\) are independent Rademacher random variables,
and set \(a_{I,\ell}:=Q_IH_Ke_\ell\).  For fixed \(x,I,J\),
\begin{equation}
  Q_IH_KS_1Z_{x,J}
  =\sum_{\ell=1}^K
    \eta_\ell a_{I,\ell}z_{x,J,\ell}^\dagger.
  \label{eq:mux-matrix-series}
\end{equation}
The two variance matrices of this Rademacher series satisfy
\begin{align}
  \sum_\ell
    \norm{z_{x,J,\ell}}^2
    a_{I,\ell}a_{I,\ell}^\dagger
    &\preceq\mu_{\mathrm{mux}}I_h,
    \label{eq:mux-left-variance}\\
  \sum_\ell
    \norm{a_{I,\ell}}^2
    z_{x,J,\ell}z_{x,J,\ell}^\dagger
    &=\frac hK I_h.
    \label{eq:mux-right-variance}
\end{align}
Lemma~\ref{lem:rademacher}, with
\(t=16\sqrt{h/K}\ln(2MK)\), followed by a union bound over the
\(MD^2\) triples \(x,I,J\), gives a total failure probability at most
\begin{equation}
  2hMD^2(2MK)^{-8}<1.
  \label{eq:mux-block-union-bound}
\end{equation}
Thus one choice of \(S_1\) satisfies
\[
  \norm{Q_IV_xQ_J^\dagger}
  \le16\sqrt{\frac hK}\ln(2MK)
  =C_{\mathrm{flat}}\sqrt{\frac hK}\,R
\]
simultaneously for all \(x,I,J\).  Combining this with the trivial
unitary-subblock bound \(\norm{Q_IV_xQ_J^\dagger}\le1\) proves the
lemma.
\end{proof}

Define
\begin{equation}
  \mathcal V:=\sum_{x=0}^{M-1}\proj{x}_{\A}\otimes V_x.
  \label{eq:mux-V}
\end{equation}
Since the Boolean phase oracles in Lemma~\ref{lem:mux-common-flatten} are common to all
branches,
\begin{equation}
  \mathcal U
  =(I_{\A}\otimes S_1H_K)\,
    \mathcal V\,
    (I_{\A}\otimes H_KS_2).
  \label{eq:mux-unflatten}
\end{equation}
The two Boolean phase oracles are implemented exactly using \(O(\sqrt K)\)
\(T\) gates and clean ancillas in total, and they occur only once,
outside the amplification sequence.

\subsection{Block encoding and synthesis}

Split the target register into a \(D\)-dimensional block label \(\B\)
and an \(h\)-dimensional register \(\Y\).  Introduce a
\(D\)-dimensional register \(\X\) and a one-qubit flag \(\f\), and define
\begin{equation}
  C_{x,IJ}:=\frac{Q_IV_xQ_J^\dagger}{g_{\mathrm{mux}}}.
  \label{eq:mux-C}
\end{equation}
Every \(C_{x,IJ}\) is a contraction.  Using the Julia--Halmos dilation
from Eq.~\eqref{eq:halmos-dilation}, define
\begin{multline}
  \SELECT_{\mathrm{mux}}
  :=\sum_{x=0}^{M-1}\sum_{I,J=0}^{D-1}
    \proj{x}_{\A}\otimes\proj{I}_{\X}\otimes\proj{J}_{\B}\\
    {}\otimes\operatorname{Hal}(C_{x,IJ})_{\f\Y}.
  \label{eq:mux-select}
\end{multline}
Let \(W_{\mathrm{mux}}\) apply \(H_D\) to \(\X\), followed by
\(\SELECT_{\mathrm{mux}}\), \(H_D\) on \(\B\), and a swap of \(\X\) and
\(\B\).  Let \(J_{0,\mathrm{mux}}\) append \(\X=0\) and \(\f=0\), while
leaving \(\A,\B,\Y\) as logical registers.  The basis-state calculation
of Lemma~\ref{lem:uniform-block} gives
\begin{equation}
  J_{0,\mathrm{mux}}^\dagger
  W_{\mathrm{mux}}J_{0,\mathrm{mux}}
  =\frac{\mathcal V}{\rho_h},
  \label{eq:mux-clean-block}
\end{equation}
where
\begin{equation}
  \rho_h:=Dg_{\mathrm{mux}}
  =\min\!\left\{
    D,C_{\mathrm{flat}}\sqrt D\,R
  \right\}.
  \label{eq:mux-rho}
\end{equation}
The normalization is independent of \(x\), and the address register is
not included in the signal projector.  Hence the same QSVT
sequence applies to arbitrary superpositions of addresses.

The \(\SELECT\) operator in Eq.~\eqref{eq:mux-select} acts on
\(N=m+2\log D+\log h+1\) input qubits, of which
\(q=\log h+1\) are target qubits.  Consequently,
\begin{equation}
  2^{(N+q)/2}=2\sqrt M\,K,
  \qquad
  4^q=4h^2.
  \label{eq:mux-compiler-factors}
\end{equation}
Define
\begin{equation}
  S_h:=\sqrt M\,K\sqrt\Lambda+h^2\Lambda.
  \label{eq:mux-Sh}
\end{equation}

\begin{theorem}
\label{thm:mux-parameterized}
For every power of two \(h\mid K\), the unitary
\(\mathcal U\) in Eq.~\eqref{eq:mux-definition} has a Clifford+\(T\)
implementation with error at most \(\eps\), \(T\)-count
\begin{equation}
  T_{\mathrm{mux}}
  =O\!\left(\rho_hS_h+\sqrt K\right),
  \label{eq:mux-parameterized-bound}
\end{equation}
and \(O(S_h+\sqrt K+\log K)\) clean ancillas.
\end{theorem}

\begin{proof}

Fix a sufficiently small constant \(c_0>0\), and compile
\(W_{\mathrm{mux}}\) to a Clifford+\(T\) unitary with error
\begin{equation}
  \delta:=\frac{c_0\eps^2}{\rho_h^2}.
  \label{eq:mux-delta}
\end{equation}
By Lemma~\ref{lem:tan-controlled} and
Eq.~\eqref{eq:mux-compiler-factors}, one compiled query has
\(T\)-count and clean-ancilla count \(O(S_h)\).  Indeed, since
\(\rho_h\le D\le K\),
\begin{align}
  \log(1/\delta)
  &= \log(1/c_0)+2\log\rho_h+2\log(1/\eps)
  \notag\\
  &= O\!\left(\log K+\log(1/\eps)\right).
\end{align}
This logarithmic factor is absorbed into \(\Lambda\).  Moreover, the same compression
argument as in Lemma~\ref{lem:compiled-base} shows that the clean block
\(A\) of the compiled query satisfies
\[
  \left\|A-\frac{\mathcal V}{\rho_h}\right\|\le\delta.
\]


Because \(W_{\mathrm{mux}}\) and \(\mathcal V\) are unitary,
Eq.~\eqref{eq:mux-clean-block} implies
\[
  \frac{1}{\rho_h}
  =
  \left\|\frac{\mathcal V}{\rho_h}\right\|
  \le 1.
\]
Thus \(1\le\rho_h\le K\).  Set \(c:=1/\rho_h\).  For sufficiently
small \(c_0\), the choice in Eq.~\eqref{eq:mux-delta} satisfies
\[
  \delta<c/2,
  \qquad
  Q^2\delta=O(c_0\eps^2)=O(1),
\]
where \(Q=\Theta(\rho_h)=\Theta(1/c)\).  Lemma~
\ref{lem:robust-amplification} therefore givesan ideal exact-phase
QSVT sequence using \(O(\rho_h)\) calls to the compiled query and its
exact inverse, with clean-input error
\[
  O(Q\sqrt{\delta})
  =O(\sqrt{c_0}\,\eps).
\]
Choosing \(c_0\) sufficiently small makes this error at most
\(2\eps/3\).  The compiled queries contribute
\(O(\rho_hS_h)\) \(T\) gates, while their \(O(S_h)\) clean workspace
is reused throughout the sequence.

By Lemma~\ref{lem:signal-phase}, the \(O(\rho_h)\) signal rotations can
be synthesized to total error at most \(\eps/3\), using
\(O(\rho_hS_h)\) \(T\) gates and \(O(S_h)\) reusable clean ancillas.
Indeed, the signal predicate involves \(O(S_h)\) qubits, and
\[
  \log(\rho_h/\eps)
  =
  O\!\left(\log K+\log(1/\eps)\right)
  =
  O(\Lambda)
  =
  O(S_h).
\]
By the triangle inequality, the resulting circuit implements
\(\mathcal V\) with error at most \(\eps\).

Finally, Eq.~\eqref{eq:mux-unflatten} adds the two Boolean phase oracles outside the QSVT sequence.  Their total cost is
\(O(\sqrt K)\) \(T\) gates and clean ancillas.  Together with the
\(O(\log K)\) clean block-label and flag qubits, this gives the stated
\(T\)-count and clean-ancilla bounds.
\end{proof}

\subsection{Optimizing the block size}

Let
\begin{equation}
  D_\star
  =K^{1/2}M^{-1/4}\Lambda^{1/4}\text{~~and~~} h_\star  =M^{1/4}K^{1/2}\Lambda^{-1/4}.
  \label{eq:mux-optimum}
\end{equation}

The next corollary shows the $T$-count of an $(m,k)$-UCU.
\begin{corollary}
\label{cor:mux-optimized}
Suppose \(1\le D_\star\le K\).  Then \(\mathcal U\) has a
Clifford+\(T\) implementation with error at most \(\eps\), \(T\)-count
\begin{equation}
  O\!\left(
    R\,M^{3/8}K^{5/4}\Lambda^{5/8}
  \right),
  \label{eq:mux-optimized-explicit}
\end{equation}
and \(O(\sqrt M\,K\sqrt\Lambda)\) clean ancillas.
\end{corollary}

\begin{proof}
Since \(\rho_h\le C_{\mathrm{flat}}\sqrt D\,R\) for every \(D\), and
\(h=K/D\), Theorem~\ref{thm:mux-parameterized} implies that the $T$-count of an $(m,k)$-UCU is
\begin{align}
O\!\Bigl(&
    R\sqrt M\,K\sqrt\Lambda\,D^{1/2}
    +RK^2\Lambda D^{-3/2}+\sqrt K
  \Bigr).
  \label{eq:mux-before-optimization}
\end{align}

Choose \(D\) to be the smallest power of two not smaller than
\(D_\star\), and set \(h=K/D\).  Because \(K\) is a power of two and
\(D_\star\le K\), this choice satisfies \(D\mid K\) and differs from
\(D_\star\) by less than a factor of two.  Substitution into
Eq.~\eqref{eq:mux-before-optimization} gives
Eq.~\eqref{eq:mux-optimized-explicit}; the additive \(O(\sqrt K)\)
term is lower order.  Moreover,
\[
  h^2\Lambda
  =O(h_\star^2\Lambda)
  =O(\sqrt M\,K\sqrt\Lambda),
\]
so the clean-ancilla bound follows from
Theorem~\ref{thm:mux-parameterized}.
\end{proof}

The restriction in Corollary~\ref{cor:mux-optimized} does not affect
the following polynomial scalings.  If \(M\ge K^2\),
Eq.~\eqref{eq:direct-mux-Tan} gives
\(\widetilde O(\sqrt M\,K)\).  If \(M\le K^2\), then
\(D_\star\ge1\).  When \(D_\star\le K\),
Corollary~\ref{cor:mux-optimized} applies.  Otherwise \(D_\star>K\),
equivalently \(\Lambda>MK^2\); choosing \(D=K\) in
Eq.~\eqref{eq:mux-before-optimization} gives $T$-count
\[
  O\!\left(
    R\sqrt M\,K^{3/2}\sqrt\Lambda
    +R\sqrt K\,\Lambda
  \right)
  =\widetilde O(M^{3/8}K^{5/4}),
\]
where the last estimate uses \(MK^2<\Lambda\).

Minimizing the parameterized bound over the admissible block sizes and
comparing it with Eq.~\eqref{eq:direct-mux-Tan}, the polynomial
dependence, after suppressing factors polylogarithmic in
\(M,K,1/\eps\), is
\begin{equation}
  \begin{cases}
    \widetilde O(M^{3/8}K^{5/4}),
      & M\le K^2,\\[1mm]
    \widetilde O(\sqrt M\,K),
      & M\ge K^2.
  \end{cases}
  \label{eq:mux-piecewise}
\end{equation}
For \(M<K^2\), the first branch improves the
\(\widetilde O(K^2)\) target-side term in the direct compiler; the two
bounds agree at \(M=K^2\).

\section{Conclusion}
\label{sec:conclusion}

We have shown that an arbitrary effectively specified unitary
\(U\in\U(2^n)\) can be synthesized over Clifford+\(T\), in the regime
\(n+\log(1/\eps)\le 2^n\), with \(T\)-count
\begin{equation}
  O\!\left(n2^{5n/4}(n+\log(1/\eps))^{5/8}\right).
\end{equation}
This bound is \(o(2^{4n/3}(n+\log(1/\eps))^{2/3})\) uniformly throughout
\(n+\log(1/\eps)\le 2^n\), and therefore improves Tan's previous worst-case bound over
the entire regime.

The key idea is to treat the target unitary as a single block-encoded
object rather than decomposing it into a long product of simpler
unitaries.  Boolean sign diagonals and Walsh transforms simultaneously
flatten all blocks of the target.  Normalizing and dilating these blocks
then allows them to be organized into one jointly compiled \(\SELECT\),
whose block is proportional to the full unitary.  Quantum singular value transformation removes the resulting
normalization by mapping the common encoded singular value to one.  Balancing the
normalization overhead against the cost of compiling the selected
dilations yields the exponent \(5/4\).

The known lower bound at constant accuracy remains
\(\Omega(2^n)\)~\cite{GossetKothariWu2026}, leaving a multiplicative gap of
\(\widetilde O(2^{n/4})\).  Determining whether arbitrary unitaries can
be synthesized with \(\widetilde O(2^n)\) \(T\) gates, or establishing a
stronger lower bound, remains the central open problem.

\section*{Acknowledgments}

Wei Zi was supported by the Guangdong Provincial Quantum Science Strategic
Initiative under Grant No.~GDZX2503001. 

The authors acknowledge the use of OpenAI's GPT-5.6 Sol for author-directed exploratory discussions,
literature searches, and assistance with manuscript preparation.
All mathematical claims and the final proofs were independently verified, simplified, and rewritten by the authors, who take full responsibility for the correctness and integrity of the mathematical results.
\setlength{\emergencystretch}{3em}
\bibliography{ref}

\appendix

\section{Proof of Lemma~\ref{lem:exact-response}}
\label{append:proof-restate}
{
For a complex polynomial \(p(x)=\sum_{j=0}^k a_jx^j\in\C[x]\), define \(p^*(x):=\sum_{j=0}^k\overline{a_j}x^j\) and \(\operatorname{Re}p(x):=\sum_{j=0}^k\operatorname{Re}(a_j)x^j\).  We say that \(p\) has parity \(r\in\{0,1\}\) if \(p(-x)=(-1)^r p(x)\); parity \(0\) means even and parity \(1\) means odd.
\begin{lemma}[\cite{Gilyen2019}]\label{lem:realpoly-to-complexpoly}
    Let $p_0(x)\in\mathbb{R}[x]$ be a degree-$Q$ polynomial for some $Q\ge 1$, such that $p_0(x)$ has parity-$(Q\mod 2)$ and $|p_0(x)|\le 1$ for all $x\in[-1,1]$. Then there exists a degree-$Q$ polynomial $P(x)\in\C[x]$ satisfying $\operatorname{Re} P=p_0$, and 
    \begin{enumerate}
        \item $P$ has parity-$(Q\mod 2)$,
        \item $\forall x\in[-1,1]$: $|P(x)|\le 1$,
        \item $\forall x\in (-\infty,-1]\cup [1,\infty)$: $|P(x)|\ge 1$,
        \item if $Q$ is even, then $\forall x\in\mathbb{R}$: $P(\ii x)P^*(\ii x)\ge 1$.
    \end{enumerate}
\end{lemma}

\begin{lemma}[Restatement of Lemma \ref{lem:exact-response}]
\label{lem:exact-response-restate}
For every \(0<c\le1\), there is an odd integer $Q=\Theta(1/c)$ and a degree-$Q$ polynomial $P(x)\in\C[x]$
such that $P(x)$ satisfies the four properties in Lemma \ref{lem:realpoly-to-complexpoly} and
\begin{equation}
  P(c)=1.
\end{equation}
\end{lemma}
\begin{proof}
Let \(Q\) be the smallest positive odd integer satisfying
\begin{equation}
  \sin\!\left(\frac{\pi}{2Q}\right)\le c.
  \label{eq:Q-choice}
\end{equation}
Using $2x/\pi \le \sin x \le x$ for $x\in [0,\pi/2]$, together with the minimality of $Q$, we obtain
If \(c=1\), then \(Q=1\) and the following bounds are immediate.  Otherwise \(Q\ge3\), so minimality may be applied to the preceding positive odd integer \(Q-2\).
\begin{equation}
  \frac1c\le Q\le\frac{\pi}{2c}+2.
  \label{eq:Q-explicit-bounds}
\end{equation}
Hence \(Q=\Theta(1/c)\).

Set $\beta:=\sin(\pi/(2Q))/c\le1$,
and write $Q=2m+1$. Define a real polynomial
\begin{equation}
  p_0(x):=\sin\!\bigl(Q\arcsin(\beta x)\bigr)=(-1)^mT_Q(\beta x),
  \label{eq:p0-definition}
\end{equation}
where $T_Q$ is the degree-$Q$ Chebyshev polynomial of the first kind. Therefore, $p_0$ has parity-$(Q\mod 2)$.
Moreover, for every $x\in[-1,1]$, we have $\beta x\in[-1,1]$. Then
\begin{equation}
  |p_0(x)|\le1.
  \label{eq:p0-bounded}
\end{equation}
At the distinguished point,
\begin{align}
  p_0(c)
   &=\sin\!\left(
      Q\arcsin\!\left(\sin\frac{\pi}{2Q}\right)
     \right)
    =1.
  \label{eq:p0-at-c}
\end{align}
By Lemma~\ref{lem:realpoly-to-complexpoly}, there exists a degree-\(Q\) polynomial \(P(x)\in\C[x]\) such that \(P\) satisfies the four properties in that lemma and
\begin{equation}
  \operatorname{Re}P(x)=p_0(x).
  \label{eq:qsp-completion-real-part}
\end{equation}
Since \(\operatorname{Re}P(c)=p_0(c)=1\) and
\(|P(c)|\le1\), we have
\begin{equation}
  P(c)=1.
  \label{eq:P-at-c-exact}
\end{equation}
\end{proof}
}

\end{document}